\documentclass[twoside,a4paper]{article}

\usepackage{amsmath,graphicx,amssymb,fancyhdr,amsthm,enumerate,textcomp} 
\usepackage[usenames]{color}
\newtheorem{thm}{Theorem}[section]

\newtheorem{prop}[thm]{Proposition}
 
\theoremstyle{definition} 
\newtheorem{defn}[thm]{Definition}
  
\theoremstyle{remark}  
  
\def\beq{\begin{eqnarray}}  
\def\eeq{\end{eqnarray}}  
\def\bsp{\begin{split}}  
\def\esp{\end{split}}

\def\d{\mathrm{d}}

\newcommand{\mf}[1]{{\mathfrak #1}}

\newcommand{\mbold}[1]{\mbox{\boldmath{\ensuremath{#1}}}}

\begin{document}   
   
\title{\Large\textbf{Pseudo-Riemannian VSI spaces II}}  
\author{{\large\textbf{Sigbj\o rn Hervik }    }
 \vspace{0.3cm} \\     
Faculty of Science and Technology,\\     
 University of Stavanger,\\  N-4036 Stavanger, Norway         
\vspace{0.3cm} \\      
\texttt{sigbjorn.hervik@uis.no} }     
\date{\today}     
\maketitle   
\pagestyle{fancy}   
\fancyhead{} 
\fancyhead[EC]{S. Hervik}   
\fancyhead[EL,OR]{\thepage}   
\fancyhead[OC]{Pseudo-Riemannian VSI metrics II}   
\fancyfoot{} 
   
\begin{abstract}  
In this paper we consider pseudo-Riemannian spaces of 
 arbitrary signature for which all of the polynomial curvature 
 invariants vanish (VSI spaces).  Using an algebraic classification 
 of pseudo-Riemannian spaces in terms of the boost-weight decomposition 
 we first show more generally  that a space which is not characterised by its invariants must possess the  ${\bf S}_1^G$-property.  As a corollary, we then show that a VSI space must possess the ${\bf N}^G$-property (these results are the analogues of the alignment theorem, including corollaries, for Lorentzian spacetimes). As an application we classify all 4D neutral VSI spaces and show that these belong to one of two classes: (1) those that possess a geodesic, expansion-free, 
 shear-free, and twist-free null-congruence (Kundt metrics), or (2) those that possess an invariant null plane (Walker metrics).  By explicit construction we show that the latter class contains a set of VSI metrics which have not previously been considered in the literature. 
 \end{abstract} 
 
\section{Introduction} 
 
In this paper we will consider an arbitrary-dimensional pseudo-Riemannian space of signature $(k,k+m)$. 
We will investigate when such a space has a degenerate curvature structure; 
in particular, we shall determine criteria for when a space, or tensor, has all vanishing polynomial curvature invariants (VSI space).  
Recall that a polynomial curvature invariant is defined as the polynomial invariants of the components of the curvature tensors.  
Previously, the VSI spaces for Lorentzian metrics have been studied \cite{VSI} and it was shown that 
these comprise a subclass of the degenerate Kundt metrics \cite{degen}.  Here, we will see that  
Kundt-like metrics also play a similar 
role for pseudo-Riemannian VSI metrics of arbitrary signature, however, we will see that another class of metrics arises in the pseudo-Riemannian case, namely the Walker metrics \cite{Walker}. 
In order to obtain these results we will utilise invariant theory to obtain important properties of the structure of tensors having degenerate invariants. In particular, tensors not characterised by their invariants will be shown to possess the ${\bf S}^G_1$-property, while in the VSI case they necessarily must possess the ${\bf N}^G$-property.  We will  
use this fact  to construct a new set 
of 4-dimensional Walker metrics with vanishing curvature invariants of neutral signature. 
 
Walker metrics are metrics possessing a invariant null-plane and have been studied in various contexts \cite{Walker,Law}. Here we will show that they also play a role in the classification of VSI metrics. Indeed, we will give a new class of VSI metrics which has not been considered before. These metrics are related to a bigger class of Walker metrics with a degenerate curvature structure. The curvature structure of these metrics are distinct from the Kundt metrics known from the Lorentzian case. One of the consequences of this district feature is that we need to consider invariants containing up to four derivatives. Indeed, interestingly, there is a family of Walker metrics which is VSI$_3$, but not VSI$_4$: the perhaps simplest member of this family is
\beq\label{VSI3}
ds^2=2du(dv+Vdu)+2dU(dV+av^4dU),    
\eeq
where $a$ is a constant. 
This peculiar property of being VSI$_3$ but not VSI$_4$  has no analogue in the Lorentzian case\footnote{In the Lorentzian case, VSI$_2$ implies VSI \cite{VSI}, while VSI$_1$ Kundt implies VSI \cite{degen}.}.

First we will review some of the techniques used in this paper. Then we will provide with the general result for tensors (or spaces) not being characterised by its invariants. This result is the analogue of the alignment theorem in the Lorentzian-signature case \cite{alignment}. Then, as a corollary, we will state the important VSI case. We will then use this VSI result to consider the 4-dimensional neutral case in detail. 

\subsection{Boost weight decomposition}  
Let us first review the boost weight classification, 
originally used to study degenerate metrics in Lorentzian geometry \cite{class}, 
in the pseudo-Riemannian case \cite{bw-pseudo}.  We will assume the manifold is of dimension $(2k+m)$ and of signature  $(k,k+m)$.
We first introduce a suitable (real) null-frame such that the metric can be written as: 
\beq 
\d s^2=2\left({\mbold\ell}^1{\mbold n}^1+\dots+{\mbold\ell}^I{\mbold n}^I+\dots+{\mbold\ell}^{k}{\mbold n}^{k}\right)+\delta_{ij}{\mbold m}^i{\mbold m}^j, 
\label{pseudometric}\eeq 
where the indices $i=1,\dots, m$.

Let us consider the $k$ independent boosts which forms an abelian subgroup of the group $SO(k,k+m)$: 
\beq 
({\mbold \ell}^1,{\mbold n}^1)&\mapsto& (e^{\lambda_1}{\mbold\ell}^1,e^{-\lambda_1}{\mbold n}^1)\nonumber\\ 
({\mbold \ell}^2,{\mbold n}^2)&\mapsto& (e^{\lambda_2}{\mbold\ell}^2,e^{-\lambda_2}{\mbold n}^2)\nonumber\\ 
& \vdots & \nonumber\\ 
({\mbold\ell}^{k},{\mbold n}^{k})&\mapsto& (e^{\lambda_k}{\mbold\ell}^{k},e^{-\lambda_k}{\mbold n}^{k}). 
\label{eq:boost}\eeq 
This action will be considered pointwise at the manifold.

For a tensor $T$, we can then consider the boost weights of this tensor, ${\bf b}\in \mathbb{Z}^k$, as follows. If we consider the components of $T$ with respect to the above-mentioned null-frame then if a component $T_{\mu_1...\mu_n}$ transforms as: 
\[ 
T_{\mu_1...\mu_n}\mapsto e^{(b_1\lambda_1+b_2\lambda_2+...+b_k\lambda_k)}T_{\mu_1...\mu_n}, 
\] 
then we will say the component $T_{\mu_1...\mu_n}$ is  
of boost weight ${\bf b}\equiv (b_1,b_2,...,b_k)$. We can now decompose a tensor  
into boost weights; in particular,  
\[ T=\sum_{{\bf b}\in  \mathbb{Z}^k}(T)_{\bf b},\]  
where $(T)_{\bf b}$ means the projection onto the components of boost weight ${\bf b}$.  
The projections $(T)_{\bf b}$ are the eigentensors of a set of commuting operators (the infinitesimal generators of the boosts) with integer eigenvalues.  
For example, a tensor $P=A{\mbold\ell}^I{\mbold n}^J{\mbold m}^i{\mbold m}^j$ with $I\neq J$ and $A$ is some scalar, has boost weight ${\bf b}=(b_1,...,b_k)$ where $b_I=-1$, $b_J=1$, other $b_i=0$.  Indeed, writing out a totally covariant tensor $T$ using the basis in (\ref{pseudometric}), the boost weight is given by ${\bf b}=(b_I)$ where $b_I=\#({\mbold n}^I)-\#({\mbold \ell}^I)$.

By considering tensor products, the boost weights obey the following additive rule:  
\beq 
(T \otimes S)_{{\bf b}}=\sum_{\tilde{\bf b}+\hat{\bf b}={\bf b}}(T)_{\tilde{\bf b}}\otimes (S)_{\hat{\bf b}}. 
\label{bsum}\eeq 
 We also note that the metric $g$ is of boost weight 0, i.e., $g=(g)_0$; hence, raising and lowering indices of a tensor do not change the boost weights. 
 
\subsection{The ${\bf S}_i$- and ${\bf N}$-properties} 
 
Let us consider a tensor, $T$, and list a few conditions that the tensor components may fulfill \cite{bw-pseudo,VSI-pseudo}: 
\begin{defn} \label{cond}We define the following conditions: 
\begin{enumerate} 
\item[B1)]{} $(T)_{\bf b}=0$, for all ${\bf b}=(b_1,b_2,b_3,...,b_k)$, $b_1>0$.  
\item[B2)]{} $(T)_{\bf b}=0$, for all ${\bf b}=(0,b_2,b_3,...,b_k)$, $b_2>0$.  
\item[B3)]{} $(T)_{\bf b}=0$, for all ${\bf b}=(0,0,b_3,...,b_k)$, $b_3>0$. 
\item[$\vdots$]{}  
\item[B$k$)]{}  $(T)_{\bf b}=0$, for all ${\bf b}=(0,0,...,0,b_k)$, $b_k>0$. 
\end{enumerate} 
\end{defn} 
 
\begin{defn} 
We will say that a tensor $T$ possesses the ${\bf S}_1$-property if and only if there exists a null frame such that condition B1) above is satisfied. Furthermore, we say that $T$ possesses the ${\bf S}_i$-property if and only if there exists a null frame such that conditions B1)-B$i$) above are satisfied. 
\end{defn} 
\begin{defn} 
We will say that a tensor $T$ possesses the ${\bf N}$-property if and only if there exists a null frame such that conditions B1)-B$k$) in definition \ref{cond} are satisfied, \emph{and}  
\[ (T)_{\bf b}=0, \text{ for }  {\bf b}=(0,0,...,0,0).\]  
 
\end{defn} 
Let us also recall the following result \cite{bw-pseudo,VSI-pseudo}:
\begin{prop} 
For tensor products we have: 
\begin{enumerate} 
\item{} 
Let $T$ and $S$ possess the ${\bf S}_i$- and ${\bf S}_j$-property, respectively. Assuming,  
with no loss of generality, that $i\leq j$, then $T\otimes S$ possesses the ${\bf S}_{i}$-property. 
\item{} Let  $T$ and $S$ possess the ${\bf S}_i$- and ${\bf N}$-property, respectively. Then  $T\otimes S$  possesses the ${\bf S}_i$-property. If $i=k$, then   $T\otimes S$ possesses the ${\bf N}$-property. 
\item{}  Let  $T$ and $S$ both possess the ${\bf N}$-property. Then  $T\otimes S$, and any contraction thereof, possesses the ${\bf N}$-property. 
\end{enumerate} 
\end{prop} 

We extend this and define a set of related conditions which will prove useful to us.  Consider a tensor, $T$, that does not necessarily 
meet any of the conditions above.  However, since the boost weights ${\bf b}\in \mathbb{Z}^k\subset{\mathbb 
R}^k$, we can consider a linear $GL(k)$ transformation, $G:\mathbb{Z}^k\mapsto \Gamma$, where $\Gamma$ is a 
lattice in $\mathbb{R}^k$.  Now, if there exists a $G$ such that the transformed boost weights, $G{\bf b}$, 
satisfy (some) of the conditions in Def.\ref{cond}, we will say, correspondingly, that $T$ possesses the ${\bf 
S}^G_i$-property.  Similarly, for the ${\bf N}^G$-property. 
 
If we have two tensors $T$ and $S$ both possessing the ${\bf S}_i^G$-property, with the same $G$, then when we take the tensor product:  
\[ (T\otimes S)_{G{\bf b}}=\sum_{G\hat{\bf b}+G\tilde{\bf b}=G{\bf b}}(T)_{G\hat{\bf b}}\otimes(S)_{G\tilde{\bf b}}.\] 
Therefore, the tensor product will also possess the ${\bf S}_i^G$-property, with the same $G$.  This will be 
useful later when considering degenerate tensors and metrics with degenerate curvature tensors.  Note 
also that the ${\bf S}_i^G$-property reduces to the ${\bf S}_i$-property for $G=I$ (the identity). 

\subsection{Tensors not characterised by its invariants}
Another useful concept is the question when a tensor/space-time is ``characterised by its invariants''. Henceforth, by \emph{invariants} we will always mean the \emph{polynomial invariants}. Such have been discussed in several papers both in the Lorentzian case, as well as in the more general case \cite{OP,epsilon}.

We will now recall some of the definitions and concepts from invariant theory, see e.g., \cite{GW,RS,eberlein}. 
For a tensor $T$, we define the action of the semi-simple group $G=O(k,k+m)$ on the components of $T$ as follows. For simplicity, assume that the components of $T$ have been lowered: $T_{a_1...a_p}$. 
We form the  $N$-tuple consisting of the components of $T$ as $X=[T_{a_1a_2...a_{p}}]\in \mathbb{R}^N$. The action corresponds to a frame rotation and explicitly, if we consider the matrix $g=(M^a_{~b})\in O(k,k+m)$, acting as a frame rotation $g\omega=\{ M^a_{~1}{\mbold e}_a,..., M^a_{~n}{\mbold e}_a\}$, the frame rotation induces an action on $X$ through the tensor structure of the components:
\[ g(X)=\left[ M^{b_1}_{~a_1}...M^{b_{p}}_{~a_{p}}T_{b_1...b_{p}}\right]. \]
The (real) orbit $\mathcal{O}(X)$ is now defined by:
\[\mathcal{O}(X)\equiv \{ g(X)\in \mathbb{R}^N ~\big{|}~g\in O(k,k+m) \}\subset \mathbb{R}^N.\]
We can then  extend this definition to a direct sum of vectors, $T=T^{(1)}\oplus ...\oplus T^{(q)}$. The action $g(X)$ on the components are then extended through the standard direct-sum representaion of the group $G$ acting on the direct sum of tensors. 

In the case of a pseudo-Riemannian space, $T$ is a direct sum of the curvature tensors, $$T={\rm Riem}\oplus \nabla{\rm Riem}\oplus \nabla\nabla{\rm Riem}\oplus...\oplus \nabla^{(K)}{\rm Riem}$$ up to some sufficiently high order $K$.

\begin{defn}
A tensor $T$ (or pseudo-Riemannian space) is \emph{characterised by its invariants} if and only if the corresponding orbit $\mathcal{O}(X)$ is topologically closed in $\mathbb{R}^N$ with respect to the standard Euclidean topology.
\end{defn} 
The motivation for this definition is given in \cite{alignment} -- essentially, the set of closed orbits:
\[
\mf{C}=\{\mathcal{O}(X)\subset V ~\big{|}~ \mathcal{O}(X) \text{ closed}  \},
\]
is parameterised by the invariants, possibly up to a complex rotation (indeed, the complexified orbits are parameterised uniquely, the real orbits intersect these a finite number of times).\footnote{In \cite{RS} they denote  this set as $V//G$.}

For more on these issues we would refer the reader to \cite{GW,RS,eberlein,alignment}.

\section{Pseudo-Riemannian metrics not characterised by its invariants}
A tensor, $T$, satisfying the ${\bf S}_i^G$-property or ${\bf N}^G$-property is not generically  
determined by its invariants in the sense that there may be another tensor, $T'$, with  
precisely the same invariants. The ${\bf S}^G_i$-property thus implies a certain \emph{degeneracy} in the tensor.  

 Indeed:
\begin{thm}\label{mainthm}
A tensor $T$ is not characterised by its invariants if and only if it possesses (at least) the ${\bf S}_1^G$-property.
\end{thm}
\begin{proof}
Assume that $T$ is not charaacterised by its invariants; i.e., the corresponding orbit is not closed. 
Using the results of Richardson-Slodowy \cite{RS}, there then exists a ${\mathcal X}\in \mathfrak{B}$, where $\mathfrak{B}$ is the vector subspace of the Lie algebra $\mf{s}\mf{o}(k,k+m)$ consisting of symmetric matrices (so that $\mf{s}\mf{o}(k,k+m)=\mf{B}\oplus \mf{K}$, where $\mf{K}$ is the Lie algebra of the maximal compact subgroup),  such that $\exp(\tau{\mathcal X})(T)\rightarrow p$. We note that the maximal compact subgroup of $SO(k,m+k)$ is $K\cong SO(k)\times SO(m+k)$, which we represent as $g=(g_1,g_2)\in SO(k)\times SO(m+k)$. The  ${\mathcal X}$ can be represented as:
\beq
{\mathcal X}=\begin{bmatrix} 
{\bf 0}_k & S \\
S^t & {\bf 0}_{m+k}
\end{bmatrix},
\eeq
where $S$ is an $k\times (k+m)$ matrix. The transformation, $g^{-1}{\mathcal X}g$ induces a transformation of $S$ according to $g_1^{-1}Sg_2$, $(g_1,g_2)\in SO(k)\times SO(m+k)$. Thus by the singular value decomposition we can always find a $g\in K$ such that $S$ is diagonal: $S=diag(\lambda_1,...,\lambda_k)$. This therefore corresponds to a pure boost; specifically, by applying a null-frame the ${\mathcal X}$ will be represented the boost given in eq.(\ref{eq:boost}). Henceforth, let us represent ${\mbold\lambda}=(\lambda_1,..,\lambda_k)$ as a vector. Then if the tensor $T$ is decomposed using the corresponding boost-weight components relative to the null-frame; i.e., $T=\sum_{\bf b}(T)_{\bf b}$, we can write:
\beq
\exp(\tau{\mathcal X})(T)_{\bf b}=\exp(\tau{\bf b}\cdot{\mbold \lambda})(T)_{\bf b}\eeq
where ${\bf b}\cdot{\mbold \lambda}=\sum_{i=1}^nb_i\lambda_i$. 
In the limit $\tau\rightarrow \infty$, $\exp(\tau{\mathcal X})(T)$ has to approach $p$ which is finite: hence, if $(T)_{\bf b}\neq 0$ we get the requirement  ${\bf b}\cdot{\mbold \lambda}\leq 0$. In particular, 
\beq
\exp(\tau{\bf b}\cdot{\mbold \lambda})(T)_{\bf b}&\rightarrow& (T)_{\bf b}, \quad {\bf b}\cdot{\mbold \lambda}=0, \nonumber \\
\exp(\tau{\bf b}\cdot{\mbold \lambda})(T)_{\bf b}&\rightarrow& 0, \quad {\bf b}\cdot{\mbold \lambda}<0,
\eeq
all other $(T)_{\bf b}$ must be zero:
\beq 
\label{eq:posbw=0}
(T)_{\bf b}=0, \quad {\bf b}\cdot{\mbold \lambda}>0.\eeq

Using a $G\in O(k)$ transformation in boost-weight space we can align ${\mbold\lambda}$ with the first basis vector so that $G{\mbold\lambda}=|{\mbold\lambda}|(1,0,0,...0)$. Thus the requirement eq.(\ref{eq:posbw=0}) implies that $T$ fulfills the ${\bf S}^G_1$-property.
\end{proof}

\subsection{The VSI properties} 
 
For the VSI spaces we now get an important corollary:

\begin{thm} 
For a tensor $T$ in pseudo-Riemannian space the following is equivalent: 
\begin{enumerate}
\item{} $T$ has only vanishing polynomial invariants (VSI).
\item{} Any operator constructed from $T$ (by rasing/lowering indices, contractions, and tensor products) is nilpotent.
\item{} $T$ possesses the ${\bf N}^G$-property. 
\end{enumerate} 
\end{thm} 
\label{thm:N-property}
\begin{proof} 
The proof of 1 $\Leftrightarrow $ 2 follows from \cite{OP}. Furthermore, 3 $\Rightarrow$ 1 follows from this work also. Left to prove is thus  1 $\Rightarrow$ 3. 

From the proof of theorem \ref{mainthm} we see that tensors having all vanishing invariants must either have closed orbits, or it has a limit which approaches an element in this closed orbit. We note that the zero-tensor $\tilde{T}=0$ has a closed orbit, and since the complex orbit consists of only the zero element, the zero-tensor must be the unique tensor which has closed (real) orbits. Thus, we can choose the limit in the proof to be $p=0$. This implies that eq.(\ref{eq:posbw=0}) turns into the stronger requirement:
\beq  (T)_{\bf b}=0, \quad {\bf b}\cdot{\mbold \lambda}\geq 0.
\eeq
By the same transformation matrix $G$, we write $G{\mbold\lambda}=|\lambda|(1,0,0,...,0)$ and the ${\bf N}^G$-property follows. 
\end{proof} 

\section{4D Neutral space: all VSI metrics} 
The even-dimensional case with signature $(k,k)$ (i.e., $m=0$) is called the \emph{neutral} case. 
Let us consider the 4D neutral case which is of particular interest (see, e.g., \cite{twist,AMCH}); in particular,  we will use the above theorem to find all neutral VSI spaces of dimension 4. Such spaces has been studied before, however, only spaces satisfying the ${\bf N}$-property were investigated. Although it was noted that the ${\bf N}^G$-property was sufficient for VSI this possibility was not investigated in detail. Indeed, we will show that there are VSI spaces satisfying the ${\bf N}^G$-property, but not the ${\bf N}$-property thus establishing a new class of VSI spacetimes. We also derive all such metrics and show that they are all Walker metrics possessing an invariant null-plane. 

In 4D neutral signature we thus get two classes of metric, the Kundt metrics and the Walker metrics. These will be reviewed in what follows. We will also utilise the work of Law \cite{Law} where all the spin-coefficients of 4D neutral space were investigated. Using Law's notation, we adopt the sligthly modified null-frame $({\mbold\ell},{\mbold n},{\mbold m},\tilde{\mbold m})\equiv ({\mbold \ell}^1,{\mbold n}^1,{\mbold \ell}^2,-{\mbold n}^2)$ so that metric (\ref{pseudometric}) can be written:
\beq
ds^2=2{\mbold\ell}{\mbold n}-2{\mbold m}\tilde{\mbold m}. 
\eeq
In the neutral case, this frame is purely real. 
With respect to such a frame, Law defined the spin-coefficients which we will use in proving the main theorem. In \cite{Law} Law writes the spin-coefficients in terms of $\kappa,~\rho,~\sigma,~\tau, ~\epsilon, ~\alpha,~\beta,~\gamma$, and their tilded ($\tilde{\kappa}, ~\tilde{\rho}, ...$.), primed (${\kappa}', ~{\rho}', ...$.), and primed-tilded ($\tilde{\kappa}', ~\tilde{\rho}', ...$.) counterparts. All these spin-coefficients are real. For example, the covariant derivatives of the frame-vector $\ell^a$ can be written as:
\beq
\ell^b\nabla_b\ell^a&=&(\epsilon+\tilde{\epsilon})\ell^a+\tilde{\kappa}m^a+\kappa\tilde{m}^a, \nonumber \\
\tilde{m}^b\nabla_b\ell^a&=&(\alpha+\tilde{\beta})\ell^a+\tilde{\sigma}m^a+\rho\tilde{m}^a, \nonumber \\
m^b\nabla_b\ell^a&=&(\tilde{\alpha}+{\beta})\ell^a+\tilde{\rho}m^a+\sigma\tilde{m}^a, \nonumber \\
n^b\nabla_b\ell^a&=&(\gamma+\tilde{\gamma})\ell^a+\tilde{\tau}m^a+\tau\tilde{m}^a, \\
&...& \text{etc.} \nonumber 
\eeq
We refer to \cite{Law}, in particular, eqs.(2.10-2.11) therein, for details. 
\subsection{Invariant null planes: Walker metrics}  
Here, we will consider the 4D neutral spaces which possesses an invariant null-plane. Such metrics are known as \emph{Walker metrics}. 

Consider two orthogonal null-vectors ${\mbold\ell}$ and ${\mbold m}$. These span an invariant null-plane iff
\beq
\nabla_a({\mbold\ell}\wedge{\mbold m})=k_a({\mbold\ell}\wedge{\mbold m}), 
\eeq
for a vector $k_a$. Using \cite{Law} this immediately implies the vanishing of certain spin-coefficients : 
\[ \kappa=\rho=\sigma=\tau=0.\] 
Indeed, one can see that the vanishing of these spin-coefficients imply the existence of an invariant null-plane (hence, it is a Walker metric).

 Furthermore, Walker \cite{Walker}  showed that the requirement of an  
invariant $2$-dimensional null plane implies that the (Walker) metric can be written in the canonical form:  
\beq 
\mathrm{d}s^2=2\d u(\d v+A\d u+C\d U)+2\d U(\d V+B\d U), 
\eeq 
where $A$, $B$ and $C$ are functions that  may depend on all of the coordinates.
 
In particular, this implies that we can choose a frame such that \cite{Law}
\beq
 \kappa=\rho=\sigma=\tau=\epsilon=\beta=0, &&\quad \alpha'=\gamma'=\rho'=\tau'=0, \\
\tilde{\kappa}=\tilde{\rho}=\tilde{\alpha}=\tilde{\epsilon}=0, && \quad \tilde{\beta}'=\tilde{\gamma}'=\tilde{\sigma}'=\tilde{\tau}'=0. 
\eeq
 We note that $\tilde{\sigma}$ needs not be zero, and hence, these Walker metrics need not be Kundt spacetimes (see below).

\subsection{Pseudo-Riemannian Kundt metrics}

In the Lorentzian case  the Kundt metrics play an important role 
for degenerate metrics, and VSI metrics in particular \cite{VSI}. Their pseudo-Riemannian analogues also play an important role for 
pseudo-Riemannian spaces of arbitary signature \cite{VSI-pseudo,AMCH}. 
 
We define the pseudo-Riemannian Kundt metrics in a similar fashion, namely:  
\begin{defn} 
A pseudo-Riemannian Kundt metric is a metric which possesses a non-zero null vector ${\mbold\ell}$ which is geodesic, expansion-free, twist-free and shear-free. 
\end{defn} 
This implies that, in terms of the spin-coefficients defined in \cite{Law}, that a space is Kundt if and only if there exists a frame such that: 
\beq
\tilde{\kappa}=\kappa=\tilde{\rho}=\rho=\tilde{\sigma}=\sigma=0. 
\eeq 

Therefore, we will consider metrics of the form (which is equivalent to the above definition) 
\beq  
\d s^2=2\d u\left[\d v+H(v,u,x^C)\d u+W_{A}(v,u,x^C)\d x^A\right]+{g}_{AB}(u,x^C)\d x^A\d x^B 
\label{Kundt} 
\eeq 
(here, the indices $A,B$ range over the null-indices $I=2,3$.
The metric (\ref{Kundt}) 
possesses a null vector field ${\mbold\ell}$ obeying\footnote{If, in addition $L_{1i}=\tilde{L}_{1i}=0$, the vector $\ell_\mu$ is also recurrent (hence, Walker), and if $L_{1i}=\tilde{L}_{1i}=L_{11}=0$, then $\ell_\mu$ is covariantly constant. } 
\[ \ell_{\mu;\nu}=L_{11}\ell_\mu\ell_\nu+L_{1i}\ell_{(\mu}m^i_{~\nu)}+\tilde{L}_{1i}\ell_{(\mu}\tilde{m}^i_{~\nu)},\] 
and consequently it is geodesic, non-expanding, shear-free and 
non-twisting. Since this is a pseudo-Riemannian space of signature $(2,2)$, then the transverse metric  
\[ \d s^2_{1}={g}_{AB}(u,x^C)\d x^A\d x^B,\]  
will be of signature $(1,1)$.

 \subsection{The 4D Neutral VSI theorem} 

Let us now state an important result regarding the deterimination of all 4D VSI metrics. 

\begin{thm}
A 4D Neutral VSI metric is of one (or both) of the following types: 
\begin{enumerate}
\item{} A Walker metric possessing an invariant 2-dimensional null-plane.  
\item{} A Kundt metric. 
\end{enumerate}
\end{thm}
In order to prove this theorem one needs to consider theorem \ref{thm:N-property} and consider the covariant derivatives $\nabla^{(N)}({\rm Riemann})$. We will prove the theorem using two different methods, one is the more indirect method using the one-parameter family of boosts $B_{\tau}=e^{\tau{\mathcal X}}$, the other is the direct method by explicitly computing the covariant derivatives. These two illustrate two conseptually different methods and both provide us with separate information about the underlying structure of these spaces. For example, while the first is a more 'elegant' proof, the second gives some information of how many derivatives are necessary and provides with more details about the various special cases. 

\paragraph{I. The boost method}
Let us employ the frame which is aligned with the family of boosts $B_{\tau}=\exp({\tau{\mathcal X}})$ providing us with the limit in Theorem \ref{mainthm}. This is a pointwise action but  consider a point $p$ and assume this is
regular\footnote{In the sense of \cite{kramer}; i.e., the number
of independent Cartan invariants do not change at $p$.} implying that there exists a neighbourhood $U$ such that the algebraic structure of the space does not change over $U$. Consider now a compact $K\subset U$ neighbourhood of $p$. The boost $B_{\tau}$ acts pointwise, however, since $K$ is compact, we can assume that the $B_{\tau}$ does not depend on the point in $K$. Thus, with respect to the adapted frame, the boost will be constant over $K$: 
\[
{\mbold\ell}\mapsto e^{-\tau\lambda_1}{\mbold\ell}, \quad {\bf
n}\mapsto e^{\tau\lambda_1}{\bf n}, \quad {\bf m}\mapsto e^{-\tau\lambda_2}{\bf m}, \quad \tilde{\bf m}\mapsto e^{\tau\lambda_2}\tilde{\bf m}
\]
Note that such a boost will transform the curvature tensors
at $p$ as follows: 
\beq
\exp(\tau{\mathcal X})(T)_{\bf b}=\exp(\tau{\bf b}\cdot{\mbold \lambda})(T)_{\bf b}.
\eeq
Now, in relation to the $\epsilon$-property \cite{epsilon}, we have that this boost manifests the limit: 
\[ X=\tilde{X}+N. \] 
Furthermore, since $K$ is compact, $||N||$ will have a maximum, $N_{\max}$, over $K$ so that $ ||N||\leq N_{\max}$; consequently, 
\[ ||X-\tilde{X}||\leq N_{\max}.\] 
In the VSI case, $\widetilde{X}=0$, so that $X=N$ and the $\epsilon$-property implies the components can be arbitrary close to flat space. 

Consider now the action of the boost $B_{\tau}$. The vector $N$ is a direct sum of tensorial objects implying that, since it must be of type III, or simpler, that there is an $a>0$ such that 
\[ ||B_{\tau}(N)||\leq e^{-a\tau}||N||\leq e^{-a\tau}N_{\max}.\]
We can assume that the neighbourhood $U$ is a coordinate patch and map $U$ into $\mathbb{R}^4$ with $p$ at the origin. Then we can assume that the compact neighbourhood $K\subset \mathbb{R}^4$. We now consider the $ X=N$ as a set of differential equations on $U$ as follows: 

Express the components of the Riemann tensor (relative to the adapted frame)  in terms of the spin-coefficients $\Gamma^{\mu}_{~\alpha\beta}$ in the standard way:
\beq
R^{\mu}_{~\alpha\beta\nu}={\mbold\partial}_{\nu}(\Gamma^{\mu}_{~\alpha\beta})-{\mbold\partial}_{\beta}(\Gamma^{\mu}_{~\alpha\nu})+\left(\Gamma\star\Gamma\right)^{\mu}_{~\alpha\beta\nu},
\eeq
where $\Gamma\star\Gamma$ indicates the quadratic terms in the spin-coefficients. Similarly, the covariant derivatives, can also be expressed using the spin-coefficients: 
\[ \nabla R=\nabla R({\mbold\partial}{\mbold\partial}\Gamma,{\mbold\partial}\Gamma,\Gamma), \quad  \nabla\nabla R=\nabla\nabla R({\mbold\partial}{\mbold\partial}{\mbold\partial}\Gamma,{\mbold\partial}{\mbold\partial}\Gamma,{\mbold\partial}\Gamma,\Gamma), \quad \text{etc.}\]
We thus replace the left-hand side of $X=N$ with a PDE: 
\beq 
\text{\sc{Pde}}[\Gamma]=N
\label{pde}.\eeq 
The relation between the frame  ${\mbold\partial}_\alpha$  and $\Gamma$ are given via:
\beq 
[{\mbold\partial}_{\alpha},{\mbold\partial}_{\beta}]=-(\Gamma^\mu_{~\alpha\beta}-\Gamma^\mu_{~\beta\alpha}){\mbold\partial}_{\mu}
\label{gamma}.\eeq
The eqs.(\ref{pde}) and (\ref{gamma}) provide us with a set of PDEs and integrability conditions over the neighbourhood $U$ in terms of the functions $\Gamma^\mu_{~\alpha\beta}$. We can now consider the ``boosted'' set of equations 
\beq 
\text{\sc{Pde}}[\widehat{\Gamma}]=B_{\tau}(N)
\label{pdeB}\eeq
over $U$. This gives us a one-parameter family of equations. Since $B_{\tau}(N)$ can be made arbitrary small, this can be seen as a pertubation of a PDE describing flat space. 
Let us now consider the Cartan equivalence problem \cite{KN} which will give us a more direct perturbation. Let us make sure we consider sufficient number of derivatives in $X$ to satify the Cartan bound. Consider the point $p$. For every $\tau$ there is an inverse boost so that the $B_{\tau}(N)$ is mapped onto $X=N$. Considering the boost that leaves the point $p$ fixed, then the equivalence principle implies that there exists a diffeomorphsim $\phi_{\tau}$ that maps $K$ onto $\phi_{\tau}(K)$, leaving $p$ fixed, and induces (through $\phi^*_{\tau}$) the boost $B_{\tau}$ acting on the tangent space at $p$. The diffeomorphism does not necessarily map $K$ into itself. Consider an increasing sequence $\tau_n$  such that $\tau_n\rightarrow \infty$,
and define $K_n=\phi_{\tau_n}(K)$, which is compact. In particular, $K_n$ is closed and $p\in K_n$. This implies further that $p\in K\cap\left(\bigcap_{n}K_n\right)$ (and closed). 

Note that the set $ K\cap\left(\bigcap_{n}K_n\right)$ may not be a neighbourhood, indeed, in many cases it may be a single line. Thus, the limiting procedure may result in a mere pointwise result at $p$ causing the functions $\widehat{\Gamma}$ to not necessarily have the right functional dependence in the limit $\tau\rightarrow \infty$ over $K$. Thus in the limit we should only consider the value of $\Gamma$ restricted to the set  $K\cap\left(\bigcap_{n}K_n\right)$. On the other hand, for $\tau_n$ finite, the result applies to a neighbourhood. 

It is thus more appropriate to consider the following perturbed PDE:
\beq
\text{\sc{Pde}}[\widehat{\Gamma}]=B_{\tau}(\phi^*_{\tau}(N))
\eeq
where $\phi^*(N)$ should be thought of as acting on the components of $N$ as functions; i.e., if $N_{a...b}$ is a component, then $\phi^*(N_{a...b})=N_{a...b}\circ\phi$ \cite{KN}.

Assuming we are considering a certain metric $g$, we know that there exists a  set of equations to this PDE. In particular, there is a continuous family of solutions $\widehat{\Gamma}(\tau)$ which solves eq.(\ref{pdeB}). Moreover, over the compact region $K_n$, since this is a perturbed PDF implies that it satisfies a Cauchy property, namely, there exists an increasing sequence $\tau_n\rightarrow \infty$, such that for any $\epsilon>0$, there exists an $M$ such that:
\beq 
n,m\geq M ~\Rightarrow ~||\widehat{\Gamma}(\tau_n)-\widehat{\Gamma}(\tau_m)||<\epsilon.
\label{Cauchy}\eeq 

The diffeomorphsm $\phi_\tau$ acts as follows on the connection
\cite{kramer,KN}: if ${\mbold\Omega}$ is the connection
form, then $\tilde{\phi}_{\tau}^*{\mbold\Omega}=\widehat{\mbold\Omega}$, where
$\tilde{\phi}_{\tau}$ is the induced transformation on the frame
bundle and $\widehat{\mbold\Omega}$ is the transformed connection, we get over $U$:
\[ \widehat{\Gamma}^{\mu}_{~\alpha\beta}=(M^{-1})^\mu_{~\nu}\left[M^\gamma_{~\alpha}\phi_t^*(\Gamma^{\nu}_{~\gamma\delta})+M^\gamma_{~\alpha,\delta}\right] M^\delta_{~\beta}.\]
Furthermore, since $p=\phi_{\tau}(p)$,  we  have
$\Gamma^{\mu}_{~\gamma\delta}=\phi_{\tau}^*(\Gamma^{\nu}_{~\gamma\delta})$
at $p$. Moreover, in the aforementioned frame, we have
$M^1_{~1,\mu}=-M^2_{~2,\mu}$, $M^3_{~3,\mu}=-M^4_{~4,\mu}$  while all other components of
$M^\gamma_{~\alpha,\delta}$ are zero. 

Eq. (\ref{Cauchy}) implies that the connection coefficients can be chosen to be arbitrary close to flat space. Component-wise we have $|\widehat{\Gamma}^\alpha_{~\beta\gamma}(\tau_n)-\widehat{\Gamma}^\alpha_{~\beta\gamma}(
\tau_m)|< \epsilon$. Since some of the components of the connection transforms as tensor components under the boost, if the component has boost-weight ${\bf b}$, we get:
\beq 
|\widehat{\Gamma}^\alpha_{~\beta\gamma}(\tau_n)-\widehat{\Gamma}^\alpha_{~\beta\gamma}(\tau_m)|
 &=& |\exp[{{\bf b}\cdot{\mbold \lambda}\tau_n}]\Gamma^\alpha_{~\beta\gamma}-\exp[{{\bf b}\cdot{\mbold \lambda}\tau_m}]\Gamma^\alpha_{~\beta\gamma}| \nonumber \\
&=&
\exp[{{\bf b}\cdot{\mbold \lambda}\tau_n}]\left|{\Gamma}^\alpha_{~\beta\gamma}-\exp[{{\bf b}\cdot{\mbold \lambda}(\tau_m-\tau_n)}]{\Gamma}^\alpha_{~\beta\gamma}\right| \nonumber \\
&<& \epsilon. 
\eeq
If we fix $m$, then it it is clear that:
\[ {\bf b}\cdot{\mbold \lambda}\leq 0, \quad \text{or} \quad \Gamma^\alpha_{~\beta\gamma}=0 ~\text{for}~~  {\bf b}\cdot{\mbold \lambda}>0.\] 
This is valid for an arbitrary point $p\in U$; hence it is valid everywhere in the neighbourhood. 

We can now consider the connection coefficients that transform tensorially, and consider the various cases. By a simple geometric argument, we get:
\begin{enumerate}
\item{} $\tilde{\kappa}=\kappa=\tilde{\rho}=\rho=\tilde{\sigma}=\sigma=0$, and hence Kundt; \emph{or} 
\item{} $\tilde{\kappa}=\tilde{\rho}=\tilde{\sigma}=\tilde{\tau}=0$, and hence, a Walker space possessing an invariant null 2-plane.  
\end{enumerate} 

\paragraph{II. The direct method}
Before we embark on the direct method let us remind ourselves of some useful identities and formulae. The covariant derivative of a tensor $T$ has the formal structure: 
\beq
\nabla T=\partial T-\sum \Gamma\star T,
\eeq
where the $\partial T$ indicates the partial derivative piece, and the $\Gamma\star T$ indicates the algebraic piece where $\Gamma$ are the spin-coefficients. 
Furthermore, also useful are the 2nd Bianchi identity and the generalised Ricci identity:
\beq
R_{ab(cd;e)}&=&0, \\
\left[\nabla_a,\nabla_b \right] T_{c_1...c_k}&=&\sum_{i=1}^kT_{c_1...d...c_k}R^d_{~c_iab},
\eeq
which enable us to permute covariant derivatives up to algebraic terms. We note that all the algebraic terms are of lower order in derivatives of $T$. 

Assuming that $T$ fulfills the ${\bf N}^G$-property, there are therefore two potential ways the covariant derivative $\nabla T$ of the tensor can violate the ${\bf N}^G$-property; namely, through the components of the partial derivatives propagating the components of $T$ across the ${\bf b}\cdot{\mbold \lambda}=0$ line in boost-weight space, and the algebraic terms. At every level of covariant derivatives we can thus first permute the derivatives as much as possible, and the impose the necessary conditions on the remaining components.  Thus we ensure that  the ${\bf N}^G$-property is valid at every lower derivative so that when using the Ricci identity it does not involve ${\bf N}^G$-property breaking terms through the algebraic piece.  

Let us first split the Riemann tensor into its irreducible parts $R$, $S_{ab}$, $W^+_{abcd}$ and $W^-_{abcd}$. For a VSI space $R=0$ so the trace-free Ricci tensor, $S_{ab}$ is equal to the Ricci tensor $S_{ab}=R_{ab}$.

Then consider a non-zero Ricci tensor.  By considering $R_{ac}R^c_{~b}$ or higher powers if necessary, we can assume the Ricci tensor is of the form (brackets mean symmetrisation): 
\beq
R= a{\mbold\ell}{\mbold\ell}+b({\mbold\ell}\tilde{\mbold m})+c\tilde{\mbold m}\tilde{\mbold m}.
\eeq
We need to compute the derivatives $\nabla^{(k)}R_{ab}$. The various cases depend on the components $a,b$ and $c$ and let us consider these in turn. 
\paragraph{\underline{$ac\neq 0$}}. Here, we can boost so that $a$ and $c$ are both constants. 
Computing first $\nabla_aR$, some of the  components are proportional to:
\beq
(-1,2): ~~a\tilde{\sigma}, && (0,1): ~~a\tilde{\kappa}, \nonumber \\
(1,0): ~~c\tilde{\sigma}, && (2,-1): ~~c\tilde{\kappa} \nonumber;
\eeq
consequently, by the ${\bf N}^G$-property, $\tilde{\kappa}=\tilde{\sigma}=0$. 
Computing $\nabla_b\nabla_aR$ we get similarly $\tilde{\rho}=\tilde{\tau}=0$. Thus this is a Walker space. 

\paragraph{\underline{$ab\neq 0$, $c=0$}}. Here, we can boost so that $a$ and $b$ are both constants. Considering the 1st derivative, $\nabla_aR$, we get (among others) the components: 
\beq
(-1,2): ~~a\tilde{\sigma}, && (0,1): ~~a\tilde{\kappa}, \nonumber \\
(0,1): ~~b\tilde{\sigma}, && (1,0): ~~b\tilde{\kappa}, ~~ (0,-1):~~b\tilde{\rho} \nonumber;
\eeq
hence, there are two possibilities $\tilde{\kappa}=\tilde{\rho}=0$, or  $\tilde{\kappa}=\tilde{\sigma}=0$. 
By computing $\nabla_b\nabla_aR$, we quickly get $\tilde{\rho}=0$. Thus we need to consider the two cases $\tilde{\sigma}\neq 0$, and $\tilde{\sigma}= 0$. 

From the 2nd derivative, and the Law's eq. (3.4) in \cite{Law}, we get the conditions: 
\beq
\tilde{\sigma}\rho=\tilde{\sigma}\sigma=\tilde{\sigma}\kappa=\tilde{\tau}\kappa=\tau\tilde{\sigma}+\rho\tilde{\tau}=0.
\eeq
If $\tilde{\sigma}\neq 0$, then $\rho=\kappa=\sigma=\tau=0$, and consequently Walker. 

Assume then $\tilde{\sigma}= 0$. If $\tilde{\tau}=0$, then the space is again Walker. Left to consider is therefore $\tilde{\tau}\neq 0$ and $\tilde{\kappa}=\tilde{\sigma}=\tilde{\rho}=0$. From the equations above, we thus get $\kappa=\rho=0$ also. If $\sigma=0$, then the space is Kundt. We need thus to check if $\sigma\neq 0$. By computing $\nabla^{(3)}R$ and $\nabla^{(4)}R$ we get numerous constraints from the requiring the ${\bf N}^G$-property. Most of these are the same as the Bianchi identity. Imposing these and some algebraic conditions on the spin-coefficients, we get the following b.w. $(0,0)$-component to be:
\[ R_{22;4311}=12\tilde{\tau}^3\sigma b.\]
By the ${\bf N}^G$-property this component has to vanish which is contradictory to the assumptions given above. Hence, the space has to be either Walker or Kundt. 

\paragraph{\underline{$b\neq 0, a=c=0$}} Here, we notice that there is a discrete symmetry which flips boost-weight space with respect to the line $b_1-b_2=0$. Using this symmetry, the case here essentially reduces to the case $ab\neq 0$ above. Thus also here the ${\bf N}^G$-property implies Walker or Kundt. 

\paragraph{\underline{$a\neq 0, b=c=0$}}. Lastly we need to consider the case when only $a$ is non-zero. First we look at $\nabla^{(2)}R$. Using the symmetry $(b_1,b_2)\mapsto (b_1,-b_2)$ we get the conditions: 
\beq
\tilde{\kappa}=\tilde{\sigma}=\tilde{\rho}=\rho=0.
\eeq
In addition the vanishing of the $(0,0)$ components implies $\kappa\tilde{\tau}=0$. If $\tilde{\tau}=0$, then the space is Walker. Assume thus $\tilde{\tau}\neq 0$, implying $\kappa=0$.

In addition, the Bianchi identities need to be fulfilled. Imposing these and computing the symmetric 2-tensor $\Box R_{ab}$, we note that this is of the following form: 
\beq
\Box R= A{\mbold\ell}{\mbold\ell}+B({\mbold\ell}\tilde{\mbold m})+C\tilde{\mbold \ell}{\mbold m}.
\eeq
If $B$ or $C$ is non-zero, then the previous computations implies that, by considering possibly 4 more derivatives, that its Walker or Kundt. The requirements $B=C=0$ impose additional conditions on the spin-coefficients. Eventually, after possibly 4 more derivatives, also this implies its Walker or Kundt.  

\paragraph{The Weyl tensor}

Let us now consider the self-dual (or anti-self-dual by orientation reversion) Weyl tensor.  This needs to be of type III, N, or O, see \cite{bw-pseudo}. If it is of type III, then $({\sf W}^+)^2$ as a bivector operator, is of type N. Consider thus the case of type N. By discrete symmetries, we can thus assume that (in a short hand notation): 
\beq
W^+=\phi ({\mbold\ell}\wedge {\mbold m})({\mbold\ell}\wedge {\mbold m}).
\eeq
We note that the discrete symmetry that act on boost-weight space as $(b_1,b_2)\mapsto (-b_2 ,-b_1)$, leaves $W^+$ invariant. 
By computing the second covariant derivative, $\nabla_b\nabla_aW^+$, we pick out the following components (including their boost-weights): 
\beq
{\mbold n}_a{\mbold n}_b(\tilde{\mbold m}\wedge{\mbold m})(\tilde{\mbold m}\wedge{\mbold m}): &\propto \kappa^2, & (2,0) \nonumber \\
\tilde{\mbold m}_a\tilde{\mbold m}_b(\tilde{\mbold m}\wedge{\mbold m})({\mbold \ell}\wedge{\mbold n}): &\propto \sigma^2, & (0,-2) \nonumber \\
{\mbold m}_a{\mbold m}_b(\tilde{\mbold m}\wedge{\mbold m})(\tilde{\mbold m}\wedge{\mbold m}): &\propto \rho^2,& (0,2) \nonumber 
\eeq
By the ${\bf N}^G$-property of $W^+$ and $\nabla^{(2)}W^+$, and using the remaining discrete symmetry we thus get the two cases: 
\[ \kappa=\sigma=0, \quad \text{or}\quad \kappa=\rho=0.\] 

Consider first $\kappa=\sigma=0$. Computing $\nabla_d\nabla_c \nabla_b\nabla_aW^+$, in particular the component ${\mbold m}_a{\mbold m}_b{\mbold m}_c{\mbold m}_d(\tilde{\mbold m}\wedge{\mbold n})(\tilde{\mbold m}\wedge{\mbold n})\propto \rho^4$ of boost weight $(2,2)$. Again, utilising the remaining discrete symmetry this must be zero. Thus, $\kappa=\sigma=\rho=0$. 

Hence, we are left with $\kappa=\rho=0$, while $\sigma$ need not be zero. Assume thus that $\sigma\neq 0$. Using the 2nd derivative once again, but this time the components: 
\beq 
\tilde{\mbold m}_a{\mbold n}_b({\mbold m}\wedge{\mbold n})({\mbold \ell}\wedge{\mbold m}) \propto &\tilde{\kappa}\sigma, & (1,1) \nonumber \\
\tilde{\mbold m}_a\tilde{\mbold m}_b({\mbold m}\wedge{\mbold n})({\mbold \ell}\wedge{\mbold m}) \propto & \tilde{\rho}\sigma, & (0,0)
\eeq 
thus, $\tilde{\kappa}=\tilde{\rho}=0$. From Law's eq.(3.4a) in \cite{Law}, it now implies  that $\tilde{\sigma}\sigma=0$; hence, $\tilde{\sigma}=0$. 

Thus we are in the situation where we are of one of the following cases:
\begin{enumerate}
\item{} $\kappa=\tilde{\kappa}=\rho=\tilde{\rho}=\tilde{\sigma}=0$, $\sigma\neq 0$. 
\item{} $\kappa=\rho=\sigma=0$. 
\end{enumerate}
It is important here that we keep track of the components of the lower derivatives. 

Consider next the first case where $\sigma\neq 0$. Then using the 4th derivative, we get the component: 
\[ \tilde{\mbold m}_a{\mbold m}_b{\mbold \ell}_c\tilde{\mbold m}_d(\tilde{\mbold m}\wedge{\mbold n})(\tilde{\mbold m}\wedge{\mbold m})\propto \sigma^2\tilde{\tau}^2, \]
of boost weight (0,0); consequently, $\tilde{\tau}=0$ and thus all the tilded variables $\tilde{\kappa}=\tilde{\rho}=\tilde{\sigma}=\tilde{\tau}=0$, and this is thus a Walker space. 

We are left to consider the second case where $\kappa=\rho=\sigma=0$. If $\tau=0$, we have a Walker space. Assume thus that $\tau\neq 0$. By computing the 4th derivative, we notice that one of the components, 
\[ \tilde{\mbold m}_a\tilde{\mbold m}_b{\mbold \ell}_c{\mbold \ell}_d({\mbold \ell}\wedge{\mbold n})({\mbold \ell}\wedge{\mbold n})\propto \tau^4.\]
This component has boost-weight (-2,-2) and has the same boost-weight as $W^+$ under the exchange of tilded spin-coefficinents with non-tilded ones. After a lengthy computation, sometimes needing to go to 8th order, we get that $\tilde{\kappa}=\tilde{\sigma}=\tilde{\rho}=0$ (analogously as above). Thus, implying that this is a Kundt space.  

If $W^+$ but $W^-\neq 0$, then we can consider the discrete symmetry which interchanges tilded spin-coefficients with non-tilded ones: $\tilde{x}\leftrightarrow x$, where $x$ is the spin coefficients. Then an identical computation as above implies that the space is either Walker with an invariant null 2-plane, or Kundt. The theorem follows then from these considerations. 

Although the argument involves 8th derivatives, it is suspected that the number of derivatives needed is less that this. In particular, no examples of spaces which are VSI$_k$ but not VSI$_{k+1}$ are known for $k>3$. The example eq.(\ref{VSI3}) is VSI$_3$ but not VSI$_4$,  however, this is a Walker metric which is a restricted class. This example, and an explanation of how this example can be extended to other similar examples, will be given later.  However, a question still remains: Are there examples of non-Walker metrics which are  VSI$_k$ but not VSI$_{k+1}$  for $k>3$? 

\subsection{Neutral VSI metrics} 
 
\subsubsection{4D Neutral case: Kundt metrics} 
 Using  
\beq 
\d s^2=2\left({\mbold\ell}{\mbold n}-{\mbold m}\tilde{\mbold m}\right).  
\eeq 
We will consider the pseudo-Riemannian Kundt case for which the transverse space  is 2-dimensional. Requiring the ${\bf N}$-property, this must be flat space (see, \cite{VSI-pseudo,AMCH}). Therefore, we can write:  
\[ -2{\mbold m}\tilde{\mbold m}=2\d U\d U=-\d T^2+\d X^2.\]  
There are two classes of 4D Neutral Kundt VSI metrics, they can be written \cite{VSI-pseudo,AMCH}:  
\beq 
\d s^2=2\d u\left(\d v+H\d u+W_{\mu_1}\d x^{\mu_1}\right)+2\d U\d V, 
\eeq  
where: 
\paragraph{Null case:} 
\beq 
W_{\mu_1}\d x^{\mu_1}&=& vW^{(1)}_{U}(u,U)\d U+W^{(0)}_{U}(u,U,V)\d U+W^{(0)}_{V}(u,U,V)\d V,\nonumber \\ 
H&=& v H^{(1)}(u,U,V)+H^{(0)}(u,U,V), 
\eeq 
\paragraph{Spacelike/timelike case:} 
\beq 
W_{\mu_1}\d x^{\mu_1}&=& vW^{(1)}\d X+W^{(0)}_T(u,T,X)\d T+W^{(0)}_X(u,T,X)\d X, \nonumber\\ 
H&=& \frac{v^2}8{\left(W^{(1)}\right)^2}+vH^{(1)}(u,T,X)+H^{(0)}(u,T,X), 
\eeq 
and  
\beq 
W^{(1)}=-\frac{2\epsilon}{X}, \text{ where } \epsilon=0, 1. 
\eeq 
 
We note that these possess an invariant null-line if $W^{(1)}=0$, and a 2-dimensional invariant null-plane if $W^{(0)}_{V}=0$ for the null  
case\footnote{In order for the spacelike/timelike case to possess an invariant  
null 2-plane, it needs to be a special case of the null case.}. 

\subsubsection{4D Neutral signature: Walker metrics}
This class of metrics provides us with a new set of VSI metrics which have not been considered before. This is due to the fact that these VSI metrics does not in general possess the ${\bf N}$-property, but rather the weaker requirement of the ${\bf N}^G$-property. 

Using the following Walker form,
\beq 
\mathrm{d}s^2=2\d u(\d v+A\d u+C\d U)+2\d U(\d V+B\d U), 
\label{Walker2}\eeq 
the result can be summarised in the following theorem: 
\begin{thm}
Consider the metric (\ref{Walker2}) where
\beq
A&=&vA_1(u,U)+VA_2(u,U)+A_0(u,U), \nonumber \\
B&=&VB_1(u,v,U)+B_0(u,v,U)  \nonumber \\
C&=&C_1(u,v,U)+VC_2(u,U)+C_0(u,U).
\label{cond1}\eeq
Then the following holds:
\begin{enumerate}
\item{} The metric is a VSI$_1$ space. If 
\[ A_2\frac{\partial ^2B_1}{\partial v^2}\neq 0, \quad \text{or} \quad  A_2\frac{\partial ^3C_1}{\partial v^3}\neq 0,\] 
then it is not VSI$_2$. 
\item{} If
\beq
B_1 &=& vB_{11}(u,U)+B_{10}(u,U) \nonumber \\
C_1 &=& v^2C_{12}(u,U)+vC_{11}(u,U)+C_{10}(u,U),
\label{cond2}
\eeq
then it is a VSI$_3$ space. If in addition, 
\[ A_2\frac{\partial ^4B_0}{\partial v^4}\neq 0,\] 
then it is not VSI$_4$. 
\item{} If eq. (\ref{cond2}) holds and, in addition:
\beq
B_0=v^3B_{03}(u,U)+v^2B_{02}(u,U)+vB_{01}(u,U)+B_{00}(u,U)
\label{cond3}\eeq
then the space is VSI.
\end{enumerate}
\end{thm} 
The proof is this result is partly by direct computation of the curvature tensors and requiring ${\bf N}^G$-property. Let us indicate how the proof goes and in the process we elude to how these can be generalised. 

Starting with the Walker form eq.(\ref{Walker2}) we can compute the Riemann tensor. We notice that the metric gives Riemann components in the lower triangular part of boost-weight space. Let us for short use notation such that the basis one-forms are 
\beq
\{ {\mbold\omega}^1, {\mbold\omega}^2,{\mbold\omega}^3,{\mbold\omega}^4\}=\{\d u,\d v+A\d u+C\d U,\d U, \d V+B\d U\} 
\eeq
Then a component of a tensor would have the boost-weight as follows (indices downstairs):
\[ (b_1,b_2)=(\#(1)-\#(2),\#(3)-\#(4));\] 
i.e., the component $R_{1223}$, say, will have boost-weight $(-1,1)$. 

For the Walker metric the components of interest in relation to the ${\bf N}^G$-property are: 
\beq
(2,-2): && R_{1414}=-B_{,vv},\\
(1,-1): && R_{1214}=-\tfrac 12C_{,vv}, \quad R_{1434}=-B_{,vV}\\
(0,0):&& R_{1212}=-A_{,vv},\quad R_{1234}=-\tfrac 12 C_{,vV}\quad R_{3434}=- B_{,VV} \\
(-1,1): && R_{1223}=A_{,vV}, \quad R_{2334}=\tfrac 12 C_{,VV}, \\
(-2,2): && R_{2323}=-A_{,VV}
\eeq
while $(R)_{ (b_1,b_2)}=0$ for $b_1+b_2>0$. Thus the Riemann tensor automatically satisfies the ${\bf S}_1^G$-property. In order for it to satisfy the ${\bf N}^G$-property we can set the components $(2,-2)$, $(1,-1)$ and $(0,0)$ to zero. Solving these equations gives the functional dependencies as given in (\ref{cond1}). This is thus a VSI$_0$ space. Indeed, by direct computation we note that $\nabla({\rm Riemann})$ satisfies the ${\bf N}^G$-property also, hence, it is in addition VSI$_1$. 

Assume then that (\ref{cond1}) is satisfied. Regarding the $\nabla^{(2)}({\rm Riemann})$ we note that this does not necessarily satisfy the ${\bf N}^G$-property (thus not VSI$_2$). One such nonvanishing scalar is $R_{abcd;ef}R^{abcd;ef}$. However, component-wise, the critical components are: 
\[ R_{1424;33}=A_2(B_{1})_{,vv}, \quad R_{2324;11}=\frac12 A_2\left(2(B_1)_{,vv}-(C_1)_{,vvv}\right).\] 
These give rise to the conditions mentioned and equating these to zero gives the solutions (\ref{cond2}). Satisfying eq.(\ref{cond2}) will now give the ${\bf N}^G$-property and thus VSI$_2$; indeed, VSI$_3$ by direct computation. 

Thus assume  (\ref{cond1}) and (\ref{cond2}) are satisfied. Computing $R_{abcd;efgh}R^{abcd;efgh}$ we get:
\[R_{abcd;efgh}R^{abcd;efgh}=576[(B_0)_{,vvvv}]^2A_2^4. \] 
Hence, it is not VSI$_4$ if $A_2(B_0)_{,vvvv}\neq 0$. Requiring that $(B_0)_{,vvvv}=0$, gives the solution in (\ref{cond3}) and by inspection $\nabla^{(4)}({\rm Riemann})$ satisfies the ${\bf N}^G$-property. This is sufficient for the metric to be VSI. 

We note that this proof also provides us with examples of metrics being VSI$_3$ but not VSI$_4$. For example, if  (\ref{cond1}) and (\ref{cond2}) are satisfied, but $A_2, (B_0)_{,vvvv}\neq 0$, then it is VSI$_3$ but not VSI$_4$. The example given in the introduction, eq.(\ref{VSI3}) is perhaps the simplest member of this family. 

Similarly, metrics being VSI$_1$ but not VSI$_2$ can be found analogously; as a simple set of examples of metrics of this kind:
\beq
ds^2=2du(dv+Vdu+bv^3dU)+2dU(dV+aVv^2dU),
\eeq
where $a$ and $b$ are constants, or functions depending on $(u,U)$, not both being zero.
\section{Discussion} 
In this paper we have studied pseudo-Riemannian metrics with degenerate curvature structure in the sense that they are not characterized by their polynomial curvature invariants. In particular, we related these to the ${\bf S}_1^G$-property. Specifically, we have three main results: 
\begin{enumerate}
\item{} In a pseudo-Riemannian space of arbitrary dimension and signature, a space (tensor) not characterized by its polynomial invariants possesses the ${\bf S}_1^G$-property. 
\item{} 
In the special case where the invariants vanish, the space (tensor) must possess the ${\bf N}^G$-property. 

\item{}In 4D neutral signature, a VSI space is either Kundt or a Walker space. 
\end{enumerate}
Indeed, in the latter case we constructed a new family of Walker VSI spaces. This shows that in the pseudo-Riemannian case these Walker metrics can provide new examples of metrics not being characterized by their invariants. Indeed, using the ideas given in this paper examples of VSI Walker metrics can be given in any signature $(k,k+m)$ where $k\geq 2$. As an example, the following is a neutral VSI Walker metric (with a 3D invariant null-space) in six dimensions: 
\[
ds^2=2du(dv+Vdu)+2dU(dV+\mathcal{V}dU)+2d\mathcal{U}(d\mathcal{V}+v^7d\mathcal{U}).
\] 
In future work, pseudo-Riemannian VSI metrics will be studied further and the ultimate aim is a full classification of VSI metrics in any dimension and signature. 

   
\appendix

\end{document}